\documentclass[a4paper,reqno,10pt]{amsart}

\allowdisplaybreaks[1]

\usepackage{geometry}

\usepackage{graphicx}
\usepackage{enumerate}
\usepackage{courier}
\usepackage{times}
\usepackage{amsmath,amssymb}
\usepackage{paralist}
\usepackage{color}

\theoremstyle{plain}
\newtheorem{Theorem}{Theorem}[section]
\newtheorem{Proposition}[Theorem]{Proposition}

\newtheorem{Corollary}[Theorem]{Corollary}
\newtheorem{Lemma}[Theorem]{Lemma}

\theoremstyle{definition}
\newtheorem{Definition}[Theorem]{Definition}

\theoremstyle{remark}

\title{The interplay of different metrics for the construction of constant dimension codes}

\author[S.~Kurz]{}

 \email{sascha.kurz@uni-bayreuth.de}

\newcommand{\cB}{\mathcal{B}}
\newcommand{\cD}{\mathcal{D}}
\newcommand{\cF}{\mathcal{F}}
\newcommand{\cG}{\mathcal{G}}

\newcommand{\cL}{\mathcal{L}}
\newcommand{\cM}{\mathcal{M}}
\newcommand{\cN}{\mathcal{N}}

\newcommand{\F}{\mathbb{F}}
\newcommand{\N}{\mathbb{N}}
\newcommand{\cdc}{\texttt{CDC}}
\newcommand{\FDRM}{\texttt{FDRM}}
\newcommand{\MRD}{\texttt{MRD}}
\newcommand{\cC}{\mathcal{C}}

\newcommand{\cS}{\mathcal{S}}

\newcommand{\cV}{\mathcal{V}}
\newcommand{\ds}{d_S}
\newcommand{\dham}{d_H}
\newcommand{\dr}{d_R}
\newcommand{\rb}{\color{red}{\bullet}}
\newcommand{\bb}{\color{blue}{\bullet}}
\newcommand{\wt}{\operatorname{wt}}
\newcommand{\rk}{\operatorname{rk}}
\newcommand{\supp}{\operatorname{supp}}
\def\drank{d_{\text{r}}}
\newcommand{\gaussmnum}[3]{\left[\begin{smallmatrix}{#1}\\{#2}\end{smallmatrix}\right]_{#3}}

\begin{document}
\maketitle

\centerline{\scshape Sascha Kurz}
\medskip
{\footnotesize
 \centerline{Mathematisches Institut, Universit\"at Bayreuth, D-95440 Bayreuth, Germany}
}

\begin{abstract}
  A basic problem for constant dimension codes is to determine the maximum possible size $A_q(n,d;k)$ of a set of $k$-dimensional subspaces in $\F_q^n$, called 
  codewords, such that the subspace distance satisfies $\ds(U,W):=2k-2\dim(U\cap W)\ge d$ for all pairs of different codewords $U$, $W$. Constant dimension codes 
  have applications in e.g.\ random linear network coding, cryptography, and distributed storage. Bounds for $A_q(n,d;k)$ are the topic of many recent research papers. 
  Providing a general framework we survey many of the latest constructions and show the potential for further improvements. As examples we give improved 
  constructions for the cases $A_q(10,4;5)$, $A_q(11,4;4)$, $A_q(12,6;6)$, and $A_q(15,4;4)$. We also derive general upper bounds for subcodes arising in those 
  constructions.  
\end{abstract}

\section{Introduction}
Let $\F_q$ be the finite field with $q$ elements, i.e., $q$ is a prime power. For two integers $0\le k\le n$ we denote by 
$\cG_q(n,k)$ the set of all $k$-dimensional subspaces in $\F_q^n$. The so-called subspace distance $\ds(U,W):=\dim(U)+\dim(W)-2\dim(U\cap W)=2k-2\dim(U\cap W)$ 
defines a metric on $\cG_q(n,k)$. A subset $\cC\subseteq\cG_q(n,k)$ is called a \emph{constant dimension code} (\cdc) and its elements are called codewords.  
The \emph{minimum (subspace) distance} of a {\cdc} $\cC$ is defined as $\ds(\cC)=\min\!\left\{\ds(U,W)\,:\, U,W\in \cC,U\neq W\right\}$. We call $\cC$ an $(n,M,d,k)_q$ {\cdc} 
if $\cC$ has cardinality $M$ and $\ds(\cC)\ge d$. The maximum possible cardinality of an $(n,M,d,k)_q$ {\cdc} is denoted by $A_q(n,d;k)$. We refer to 
the recurrently updated survey \cite{TableSubspacecodes} and the associated webpage \texttt{subspacecodes.uni-bayreuth.de} for some of the latest bounds. For $2k\le n$ 
and $d\ge 4$ the general bounds  
\begin{equation}
  \label{ie_cdc_lb_ub}
  q^{(n-k)\cdot (k-d/2+1)}\le A_q(n,d;k)\le 1.7314\cdot q^{(n-k)\cdot (k-d/2+1)} 
\end{equation}
are known, see \cite[Proposition 8]{heinlein2017asymptotic} for the details and further improvements depending on $q$, $k$, and $d$. For some applications 
the factor of at most $1.7314$ between the lower and upper bounds is sufficiently good. As applications are 	manifold, including e.g.\ random linear network coding, cryptography, 
and distributed storage, see e.g.\ \cite{greferath2018network}, we are interested in exact values or relatively tight bounds for $A_q(n,d;k)$ for specific, mostly small, parameters.

With respect to recent improved constructions we mention e.g.\ \cite{chen2020new,cossidente2019combining,feng2020bounds,he2020construction,he2020improving,he2021new,kurz2020lifted,
lao2020parameter,li2019construction,liu2019parallel,niu2020new}. Most of the contained improvements fit into a general framework of a combination of subcodes of a 
specific shape that we will present here. All constructions are based on an interplay between the subspace, the Hamming, and the rank metric distance.

Besides structuring and classifying the recent progress we show further potential for improvements. As examples we give improved  
constructions for the cases $A_q(10,4;5)$, $A_q(11,4;4)$, $A_q(12,6;6)$, and $A_q(15,4;4)$. Note that the dimensions of the ambient spaces are rather small. We also give 
general upper bounds for the mentioned subcodes with special shapes. 

The remaining part of this paper is structured as follows. In Section~\ref{sec_preliminaries} we introduce the necessary preliminaries and review constructions from the literature. 
The impact of codes in the Hamming metric is discussed in Subsection~\ref{subsec_skeleton_codes}. Here we generalize the notion of skeleton codes from the Echelon--Ferrers construction.  
These codes are mainly used to describe and control the combination of different subcodes to a constant dimension code. For the contained subcodes the rank metric plays an 
important role for the construction, see Subsection~\ref{subsec_rank_metric}. Based on the underlying general construction strategy sufficient conditions for adding 
further codewords are described in Subsection~\ref{subsec_additional_codewords}. In Subsection~\ref{subsec_special_constructions} we mention a few constructions outside this scheme, 
which can nevertheless be used as subcodes. At the end of Section~\ref{sec_preliminaries} the most important abbreviations and notation are listed in Table~\ref{table_notation}. 
We summarize our four exemplary improvements, that are parametric in the field size $q$, in Section~\ref{sec_improved_constructions}.    
We have chosen examples with rather small parameters and focus on the underlying techniques to show the potential for further and similar improvements for larger parameters. Upper 
bounds for the occurring subcodes are the topic of Section~\ref{sec_upper_bounds}. In Section~\ref{sec_open_problems} we recollect open problems for 
further research mentioned in Section~\ref{sec_improved_constructions} and Section~\ref{sec_upper_bounds}.

\section{Preliminaries and review of constructions from the literature}
\label{sec_preliminaries}
Let $\cC$ be a {\cdc} consisting of $k$-dimensional subspaces $U\in\cG_q(n,k)$. Given a non-degenerate bilinear form, we denote
by $U^\perp$ the orthogonal subspace of a subspace $U$, which then has dimension $n-\dim(U)$. With this, we have $\ds(U,W)=\ds(U^\perp,W^\perp)$,
so that $A_q(n,d;k)=A_q(n,d;n-k)$. Using this relation we will mostly assume $2k\le n$ in the following, so that the maximum 
possible subspace distance is $2k$. 

As a representation for a codeword $U\in\cC$ we use 
generator matrices $M\in\F_q^{k\times n}$ whose $k$ rows form a basis of $U$ and write $U=\langle M\rangle$. Applying the Gaussian 
elimination algorithm to $M$ gives a unique generator matrix $E(M)$ in \emph{reduced row echelon form}. We will also directly write $E(U)$ 
for $E(M)$ where $M$ is an arbitrary generator matrix for $U$. By $v(M)\in\F_2^n$ or $v(U)\in\F_2^n$ we denote the characteristic vector of 
the pivot columns in $E(M)$ or $E(U)$, respectively. These vectors are also called \emph{identifying} or \emph{pivot vectors}. In the following 
we will mostly use the notations $E(U)$ and $v(U)$ for $k$-dimensional subspaces of $\F_q^n$. The \emph{Ferrers tableaux} $T(U)$ of $U$ arises 
from $E(U)$ by removing the zeroes from each row of $E(U)$ left to the pivots and afterwards removing all pivot columns. If we then replace all 
remaining entries by dots we obtain the \emph{Ferrers diagram} $\cF(U)$ of $U$ which only depends on the identifying vector $v(U)$, so that we 
also write $\cF(v(U))$. As an example we consider
$$
  U=\left\langle\begin{pmatrix}
  1&0&1&1&1&0&1&0&1\\
  1&0&0&1&1&1&1&1&1\\
  0&0&0&1&0&0&0&1&0\\ 
  0&0&0&0&0&1&1&0&1
  \end{pmatrix}\right\rangle\in\cG_2(9,4),
$$  
where we have 
$$
  E(U)=
  \begin{pmatrix}
  1&0&0&0&1&0&0&0&0\\
  0&0&1&0&0&0&1&1&1\\
  0&0&0&1&0&0&0&1&0\\
  0&0&0&0&0&1&1&0&1
  \end{pmatrix},
$$
$v(U)=1	0	1	1	0	1	0	0	0\in\F_2^{9}$,   
$$
  T(U)=
  \begin{pmatrix}
   0&1&0&0&0\\
    &0&1&1&1\\
    &0&0&1&0\\
    & &1&0&1
  \end{pmatrix},
$$
and
$$
  \cF(U)=
  \begin{array}{lllll}
    \bullet & \bullet & \bullet & \bullet & \bullet \\
            & \bullet & \bullet & \bullet & \bullet \\
            & \bullet & \bullet & \bullet & \bullet \\
            &         & \bullet & \bullet & \bullet 
  \end{array}.
$$
The partially filled matrix $T(U)$ contains all essential information to describe the codeword $U$, where each entry is arbitrary in $\F_q$ and every different choice gives 
a different $k$-dimensional subspace in $\F_q^n$. The pivot vector $v(U)$ and the Ferrers diagram $\cF(U)$ of $U$ both partition $\cG_q(n,k)$ into specific classes. Note that 
this classification is not preserved by the isometries of $\F_q^n$ with respect to $\ds$. However the description with pivot vectors will be rather useful for constructions as 
we will see later on. If $n$ is given, $v(U)$ and $\cF(U)$ can be converted into each other.\footnote{The only issue occurs for pivot vectors $v(U)$ starting with a sequence of zeroes 
corresponding to the same number of leading empty columns in the Ferrers diagram. The latter, or their number, may not be directly visible.} So, we also write $v(\cF)$ for a given 
Ferrers diagram.   

\subsection{Skeleton codes, the Hamming metric, and the Echelon--Ferrers construction}
\label{subsec_skeleton_codes}
The \emph{Hamming distance} $$\dham(u,w)=\#\left\{1\le i \le n\,:\, u_i\neq w_i\right\},$$ for 
$u,w\in\F_q^n$, can be used to lower bound the subspace distance between two codewords 
$U,W\in\cG_q(n,k)$:
\begin{Lemma}(\cite[Lemma 2]{etzion2009error})\\
  \label{lemma_dist_subspace_hamming}
  For $U,W\in\cG_q(n,k)$ we have $\ds(U,W)\ge \dham(v(U),v(W))$.
\end{Lemma}  
The \emph{Hamming weight} $\wt(v)$ of a vector $v\in\F_q^n$ is its Hamming distance to the zero vector $\dham(v,\mathbf{0})$ or, in other words, the number of non-zero entries.   
If $\cS$ is a subset of $\F_2^n$ of cardinality at least $2$, then we define $\dham(\cS):=\min\{ \dham(v,v')\,:\, v,v'\in\cS, v\neq v'\}$. If $\#\cS<2$, then we formally set 
$\dham(\cS):=\infty$. We call $\dham(\cS)$ the \emph{minimum Hamming distance} of $\cS$. In applications for constant dimension codes we will assume that the elements of $\cS$ 
all have the same Hamming weight $k$. The vectors in $\F_2^n$ with Hamming weight $k$ are in one-to-one correspondence with the $k$-element subsets of an $n$-element set. So,    
slightly abusing notation, we define $\cG_1(n,k):=\left\{v\in\F_2^n\,:\,\wt(v)=k\right\}$. An $(n,M,d,k)_q$ {\cdc} $\cC$ such that all codewords have the same pivot vector $v$ 
is called $(n,M,d,k,v)_q$ {\cdc}. Directly from Lemma~\ref{lemma_dist_subspace_hamming} we can conclude:
\begin{Theorem}(\cite[Theorem 3]{etzion2009error})\\ \label{thm_EF}
  Let $\cS\subseteq\cG_1(n,k)$ with $\dham(\cS)\ge d$. If $\cC_v\subseteq\cG_q(n,k)$ is an 
  $(n,\star,d,k,v)_q$ {\cdc} for each $v\in \cS$, then $\cC=\cup_{v\in\cS} \cC_v$ is an $(n,\star,d,k)_q$ {\cdc} with cardinality $\sum_{v\in \cS} \#\cC_v$.
\end{Theorem}
Suitable choices for the $\cC_v$ are also discussed in \cite{etzion2009error} and we will do so in Subsection~\ref{subsec_rank_metric}. The underlying construction is called 
\emph{multilevel construction} in \cite{etzion2009error} and \emph{Echelon-Ferrers construction} in some other papers. Actually, the set $\cS$ is a binary code with minimum 
Hamming distance $d$ and sometimes called \emph{skeleton code}. By $A_q(n,d;k;v)$ we denote the maximum possible cardinality $M$ of an $(n,M,d,k,v)_q$ {\cdc}, so that 
Theorem~\ref{thm_EF} gives the lower bound  
\begin{equation}
  A_q(n,d;k)\ge \sum_{v\in \cS} A_q(n,d;k;v),
\end{equation}
where $\dham(\cS)\ge d$.

We can slightly generalize our notion to sets $\cV$ of binary vectors in $\F_2^n$ with Hamming weight $k$ each. If all pivot vectors of the codewords of an 
$(n,M,d,k)_q$ {\cdc} $\cC$ are contained in $\cV$, then we speak of an $(n,M,d,k,\cV)_q$ {\cdc} and denote the corresponding maximal possible cardinality by $A_q(n,d;k,\cV)$. 
For two subsets $\cV,\cV'$ of $\F_2^n$ we define their \emph{minimum Hamming distance} as $\dham(\cV,\cV'):=\min\{ \dham(v,v'),:\, v\in\cV,v'\in\cV'\}$. With this, we can directly 
generalize Theorem~\ref{thm_EF} to:
\begin{Theorem}
  \label{thm_generalized_skeleton_code}
  Let $\cV_1,\dots,\cV_s$ be subsets of $\cG_1(n,k)$ with $\dham(\cV_i,\cV_j)\ge d$ for all $1\le i<j\le s$. If $\cC_{\cV_i}\subseteq\cG_q(n,k)$ is an 
  $(n,\star,d,k,\cV_i)_q$ {\cdc} for each $1\le i\le s$, then $\cC=\cup_{1\le i\le s} \cC_{\cV_i}$ is an $(n,\star,d,k)_q$ {\cdc} with cardinality $\sum_{1\le i\le s} \#\cC_{\cV_i}$. 
\end{Theorem} 
We call $\cS=\left\{\cV_1,\dots,\cV_s\right\}$ a \emph{generalized skeleton code} and call $$\dham(\cS):=\min\left\{\dham(\cV_i,\cV_j)\,:\, 1\le i<j\le s\right\}$$ the \emph{minimum (Hamming) 
distance} of $\cS$. With this, we have the lower bound
\begin{equation}
  \label{ie_EF_gen}
  A_q(n,d;k)\ge \sum_{\cV\in \cS} A_q(n,d;k;\cV),
\end{equation}
where $\dham(\cS)\ge d$.

In several constructions in the literature, Inequality~(\ref{ie_EF_gen}) is, indirectly, applied. To this end we introduce more notation 
to describe specially structured subsets of $\cG_1(n,k)$, i.e., by
$$
  {n_1 \choose k_1},\dots, {n_l \choose k_l} 
$$
we denote the set of binary vectors which contain exactly $k_i$ ones in positions $1+\sum_{j=1}^{i-1} n_j$ to $\sum_{j=1}^{i} n_j$ for all $1\le i\le l$. 
The cases of at least $k_i$ ones are denoted by ${n_i\choose {\ge k_i}}$ and the cases of at most $k_i$ ones are denoted by ${n_i\choose {\le k_i}}$. Also in this 
generalized setting we assume that the described set is a subset of $\cG_1(n,k)$, where $n=\sum_{i=1}^l n_i$ and $k=\sum_{i=1}^l k_i$, e.g.\ 
$$
  {n_1 \choose {\le k_1}},{{n-n_1} \choose {\ge k-k_1}} \subseteq \cG_1(n,k).
$$       
In our notation, the \emph{linkage construction} from \cite[Theorem 2.3]{gluesing10construction}, \cite[Corollary 39]{silberstein2015error} can be written as
\begin{equation}
  \label{ie_linkage}
  A_q(n,d;k)\ge A_q\!\left(n,d;k;{{n-\Delta}\choose k},{\Delta\choose 0}\right)+A_q\!\left(n,d;k;{{n-\Delta}\choose 0},{\Delta\choose k}\right),
\end{equation}
which was improved to
\begin{eqnarray}
  A_q(n,d;k)&\ge& A_q\!\left(n,d;k;{{n-\Delta}\choose k},{\Delta\choose 0}\right)\notag\\ 
  &&+A_q\!\left(n,d;k;{{n-\Delta-k+d/2}\choose 0},{\Delta+k+d/2\choose k}\right)\label{ie_linkage_gen}
\end{eqnarray}
in \cite[Theorem 18, Corollary 4]{heinlein2017asymptotic}, where $0\le \Delta\le n$ is a free parameter. With respect to Inequality~(\ref{ie_linkage}) we remark 
$A_q\!\left(n,d;k;{{n-\Delta}\choose 0},{\Delta\choose k}\right)=A_q(\Delta,d;k)$ and that one key observation in \cite{gluesing10construction} was
\begin{equation}
  \label{ie_lifted}
  A_q\!\left(n,d;k;{{n-\Delta}\choose k},{\Delta\choose 0}\right)\ge q^{\Delta(k-d/2+1)}A_q(n-\Delta,d;k),
\end{equation}  
so that the two summands can be expressed in terms of $A_q(n',d;k)$ values. We will deduce Inequality~(\ref{ie_lifted}) in Subsection~\ref{subsec_rank_metric}. 
Clearly, the Hamming distance between ${{n-\Delta}\choose k},{\Delta\choose 0}$ and ${{n-\Delta}\choose 0},{\Delta\choose k}$ is $2k$, so that 
Inequality~(\ref{ie_linkage}) is a direct implication of Theorem~\ref{thm_generalized_skeleton_code} since the minimum subspace distance between two 
$k$-dimensional subspaces is at most $2k$, assuming $2k\le n$. Observing that the minimum Hamming distance between ${{n-\Delta}\choose k},{\Delta\choose 0}$ and     
${{n-\Delta-k+d/2}\choose 0},{\Delta+k+d/2\choose k}$ is at least $d$ yields Inequality~(\ref{ie_linkage_gen}).

From the computational point of view Theorem~\ref{thm_generalized_skeleton_code} translates to a weighted maximum clique problem, where the vertices are 
the candidates for $\cV\subseteq\cG_1(n,k)$ and two vertices $\cV,\cV'$ are joined by an edge iff $\dham(\cV,\cV')\ge d$. For constructive lower bounds for 
$A_q(n,d;k)$ we choose any constructive lower bound $\underline{A}_q(n,d;k;\cV)\le A_q(n,d;k;\cV)$ as vertex weights. Known upper bounds $\overline{A}_q(n,d;k;\cV)\ge 
A_q(n,d;k;\cV)$ can also be used as vertex weights. However, then the exact solution of the weighted maximum clique problem does not give an upper bound 
for $A_q(n,d;k)$ but only an upper bound on the code sizes that can be obtained by Theorem~\ref{thm_generalized_skeleton_code} using a specific generalized skeleton 
code $\cS$. Note that in principle we can choose all non-empty subsets of $\cG_1(n,k)$ as vertices. However, this set is really huge, so that 
one usually considers only suitably selected subsets thereof. For the case of $1$-element subsets of $\cG_1(n,k)$, i.e., the Echelon--Ferrers construction, 
cf.~Theorem~\ref{thm_EF}, exhaustive searches where performed in \cite{feng2020bounds}. There also upper bounds for the code sizes that can be achieved 
by the Echelon--Ferrers construction, based on Theorem~\ref{thm_upper_bound_ef} as vertex weights, were computed. While lower and upper bounds for the 
Echelon--Ferrers construction can be computed as a polynomial in the 
field size $q$, see \cite{feng2020bounds} for the details, the parametric determination of the {\lq\lq}optimal{\rq\rq} (generalized) skeleton code is a 
hard problem. So far it is only solved for the case of so-called partial spreads corresponding to $A_q(n,2k;k)$, where $n\ge 2k$, see \cite[Theorem 5.2]{feng2020bounds}. 
In our subsequent results on lower bounds for $A_q(n,d;k)$ we will always state the underlying skeleton codes. Note that the corresponding distance analysis 
in the Hamming metric, cf.~Inequality~(\ref{ie_linkage_gen}), can be parametric. To sum up, Theorem~\ref{thm_generalized_skeleton_code} is just a general 
framework for constructions and the selection of good generalized skeleton codes is a non-trivial problem. The decomposition of a given $\cdc$ $\cC$ into 
subcodes $\cC_\cV$ such that $\cC$ is given by Theorem~\ref{thm_generalized_skeleton_code} is also non-trivial, if the generalized skeleton code $\cS$ has size 
at least two, but useful indeed.          

\subsection{Vertex weights, rank-metric codes, and corresponding constructions}
\label{subsec_rank_metric}
If the pivot vectors of two codewords coincide, then we can utilize the \emph{rank distance} 
$\dr(A,B):=\operatorname{rank}(A-B)$ for matrices $A,B\in\F_q^{m\times l}$ to express the corresponding subspace distance.
\begin{Lemma}(\cite[Corollary 3]{silberstein2011large})\\[-4mm]
  \label{lemma_dist_subspace_rank}

  \noindent 
  For $U,W\in\cG_q(n,k)$ with $v(U)=v(W)$ we have $\ds(U,W)=2\dr(E(U),E(W))$.
\end{Lemma}

Since $\dr$ is a metric, we call a subset $C\subseteq\F_q^{m\times l}$ of matrices 
a \emph{rank-metric code}. If $C$ is a linear subspace of $\F_q^{m\times l}$ we call the code 
\emph{linear}. Given a Ferrers diagram $\cF$ with $m$ dots in the rightmost column and $l$ dots 
in the top row, we call a rank-metric code $C_{\cF}$ a \emph{Ferrers diagram rank-metric} (\FDRM) 
code if for any codeword $M\in \F_q^{m\times l}$ of $C_{\cF}$ all entries not in $\cF$ are zero. 
By $\dr(C_{\cF})$ we denote the minimum rank distance, i.e., the minimum of the rank distance between 
pairs of different codewords. 

\begin{Definition}\label{definition_lifted}
  (\cite{silberstein2015error})\\
  Let $\cF$ be a Ferrers diagram and $C_{\cF}\subseteq \F_q^{k\times(n-k)}$ be an {\FDRM} code. The 
  corresponding \emph{lifted {\FDRM} code} $\cC_{\cF}$ is given by
  $$
    \cC_{\cF}=\left\{U\in\cG_q(n,k)\,:\, \cF(U)=\cF, T(U)\in C_{\cF}\right\}.
  $$ 
\end{Definition}

Directly from Lemma~\ref{lemma_dist_subspace_rank} and Definition~\ref{definition_lifted} we can conclude:
\begin{Lemma}(\cite[Lemma 4]{etzion2009error})\\\label{lemma_FDRM_CDC_equivalence} 
  Let $C_{\cF}\subseteq \F_q^{k\times(n-k)}$ be an {\FDRM} code with minimum rank distance $\delta$, then the lifted 
  {\FDRM} code $\cC_{\cF}\subseteq \cG_q(n,k)$ is an $(n,\#C_{\cF},2\delta,k)_q$ {\cdc}.
\end{Lemma}

Lifted {\FDRM} codes $\cC_{\cF}$ are exactly the subcodes $\cC_v$ needed in the Echelon-Ferrers construction in Theorem~\ref{thm_EF}.
In \cite[Theorem 1]{etzion2009error} a general upper bound for (linear) {\FDRM} codes was given. Since the bound is also 
true for non-linear {\FDRM} codes, as observed by several authors, denoting the pivot vector corresponding to 
a given Ferrers diagram $\cF$ by $v(\cF)$ and using Lemma~\ref{lemma_FDRM_CDC_equivalence}, we can rewrite the upper bound to:
\begin{Theorem}
  \label{thm_upper_bound_ef}
  $$
    A_q(n,d;k;v(\cF))\le q^{\min\!\left\{\nu_i\,:\,0\le i\le d/2-1\right\}},
  $$
  where $\nu_i$ is the number of dots in $\cF$, which are neither contained in the first $i$ rows nor contained in the 
  last $\tfrac{d}{2}-1-i$ columns. 
\end{Theorem}  
If we choose a minimum subspace distance of $d=6$, then we obtain $$A_2(9,6;4;101101000)\le 2^7$$ due to
$$
  \begin{array}{lllll}
    \bb & \bb & \bb & \rb & \rb \\
        & \bb & \bb & \rb & \rb \\
        & \bb & \bb & \rb & \rb \\
        &     & \bb & \rb & \rb 
  \end{array}
  \quad  
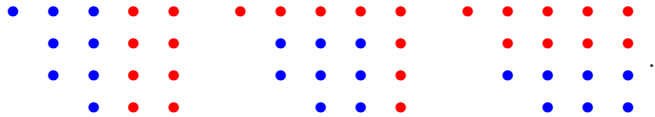
  \begin{array}{llllll}
    \rb & \rb & \rb & \rb & \rb \\
        & \bb & \bb & \bb & \rb \\
        & \bb & \bb & \bb & \rb \\
        &     & \bb & \bb & \rb 
  \end{array}
  \quad  
  \begin{array}{llllll}
    \rb & \rb & \rb & \rb & \rb \\
        & \rb & \rb & \rb & \rb \\
        & \bb & \bb & \bb & \bb \\
        &     & \bb & \bb & \bb 
  \end{array}
  .
$$
where the blue dots are those that are neither contained in the first $i$ rows nor contained in the last $\tfrac{d}{2}-1-i$ columns for $1\le i\le 3$.

While it is conjectured that the upper bound from Theorem~\ref{thm_upper_bound_ef} (and the corresponding bound for {\FDRM} codes) can always be attained, this 
problem is currently solved for specific instances like e.g.\ rank-distances $\delta=2$ only. For more results see e.g.\ \cite{antrobus2019maximal,liu2019constructions} 
and the references mentioned therein. Another important solved case are rectangular Ferrers diagrams. If $2\le 2k\le n$ and $\cF$ is the rectangular Ferrers diagrams 
with $k$ dots in each column and $n-k$ dots in each row, then a rank-metric code $C_{\cF}\subseteq\F_q^{k\times(n-k)}$ attaining the maximum possible cardinality 
$q^{(n-k)(k-d/2+1)}$ for a given minimum subspace distance $d\le 2k$ is called \emph{maximum rank distance} (\MRD) code. More generally, the maximum size of an 
$(m\times n,d_r)_q$-rank metric code is given by $m(q,m,n,\drank) := q^{\max\{m,n\}\cdot(\min\{m,n\}-\drank+1)}$. A rank metric code ${\cM} \subseteq \F_q^{m\times n}$ attaining 
this bound is said to be an {\MRD} code with parameters $(m \times n, \drank)_q$ or {\em $(m \times n, \drank)_q$ {\MRD} code}, see e.g.\ the survey \cite{sheekey2019mrd}. 
Linear {\MRD} codes exist for all parameters. Moreover, for $\drank<\drank'$ we can assume the existence of a linear $(m \times n, \drank)_q$ {\MRD} code that contains an 
$(m \times n, \drank')_q$ {\MRD} code as a subcode. The rank distribution of an additive $(m \times n, \drank)_q$ {\MRD} code is completely determined by its parameters, i.e., 
the number of codewords of rank $r$ is given by 
\begin{equation}
  \label{eq_rank_distribution_additive_MRD_code}
  a(q,m,n,\drank,r):=\gaussmnum{\min\{n,m\}}{r}{q}\sum_{s=0}^{r-d_r} (-1)^sq^{{s\choose 2}}\cdot\gaussmnum{r}{s}{q}\cdot\left(q^{\max\{n,m\}\cdot(r-d_r-s+1)}-1\right)
\end{equation}
for all $d_r\le r\le \min\{n,m\}$, see e.g.\ \cite[Theorem 5.6]{delsarte1978bilinear} or \cite[Theorem 5]{sheekey2019mrd}, where 
\begin{equation}
  \gaussmnum{n}{k}{q}=\prod_{i=0}^{k-1} \frac{q^{n-i}-1}{q^{k-i}-1}
\end{equation}
is the \emph{Gaussian binomial coefficient} counting the number of $k$-dimensional subspaces in 
$\F_q^n$. Clearly, there is a unique codeword of rank strictly smaller than $\drank$ -- the zero matrix.

Since even linear {\MRD} codes exist for all parameters, lifting gives the well-known lower bound
\begin{equation}
  A_q(n,d;k)\ge q^{(n-k)(k-d/2+1)}
\end{equation}
(assuming $2k\le n$), which is at least half the optimal value for $d\ge 4$, see e.g.\ \cite[Proposition 8]{heinlein2017asymptotic} and Inequality~(\ref{ie_cdc_lb_ub}).
In general, a subset $M\subseteq\F_q^{k\times n}$ with minimum rank distance $\delta$ is called $(k\times n,\delta)_q$-rank metric code.

Instead of starting with an {\FDRM} code $C_{\cF}$ and lifting it to a {\cdc} $\cC_{\cF}$ one can also start from an $(m,N,d,k)_q$ {\cdc} 
$\cC$ and an {\MRD} code $\cM\subseteq \F_q^{k\times(n-m)}$ with minimum rank distance $d/2$. With this we can construct a {\cdc}
\begin{equation}
  \cC'=\left\{\langle E(U)|M\rangle\,:\,U\in\cC, M\in\cM\right\}\subseteq \cG_q(n,k)
\end{equation}
with $\ds(\cC')=d$ and $\#\cC'=\#\cC\cdot\#\cM$, where $A|B$ denotes the concatenation of two matrices $A$ and $B$ with the same number of rows. This 
lifting variant was called \emph{Construction D} in \cite[Theorem 37]{silberstein2015error}, cf.\ \cite[Theorem 5.1]{gluesing9cyclic}. By construction, 
the identifying vectors of the codewords of $\cC'$ contain their $k$ ones in the first $m$ positions. Thus, we end up with Inequality~(\ref{ie_lifted}).

Lower bounds for 
$$
  A_q\!\left(n,d;k;{{n'}\choose {k'}},{{n-n'}\choose {n-k'}}\right) 
$$
were obtained in \cite{heinlein2017coset}, where the underlying construction was named \emph{coset construction}. In \cite{xu2018new} the inequality
\begin{eqnarray}
  A_q(n,d;k)&\ge& A_q\!\left(n,d;k;{{n-\Delta}\choose k},{\Delta\choose 0}\right)\notag\\
  &&             +A_q\!\left(n,d;k;{{n-\Delta}\choose {\le k-d/2}},{\Delta\choose {\ge d/2}}\right),  
\end{eqnarray}  
which holds for all $0\le\Delta<n$ due to Theorem~\ref{thm_generalized_skeleton_code}, was used in the special case $\Delta=k$ to construct many {\cdc}s 
with larger sizes than previously known. In \cite{kurz2019note} the quantity $A_q\!\left(n,d;k;{{n-\Delta}\choose {\le k-d/2}},{\Delta\choose {\ge d/2}}\right)$    
 was introduced as $B_q(n,\Delta,d;k)$. A lower bound for $A_q\!\left(n,d;k;{{n-k}\choose {\le k-d/2}},{k \choose {\ge d/2}}\right)$ was constructed in \cite{xu2018new} via
 $$
   \left\{\langle M|I_k\rangle\,:\, M\in \cM,\operatorname{rank}(M)\le k-d/2\right\},
 $$
 where $I_k$ denotes the $k\times k$ unit matrix and $\mathcal{M}\subseteq \F_q^{k\times(n-k)}$ is a rank metric code with $\dr(\cM)\ge d/2$. Note that the generator matrices $(M|I_k)$ 
 are not in reduced row echelon form in general. By replacing $I_k$ by $E(U)$ for all codewords of a $(\Delta,\star,d,k)_q$ {\cdc} we obtain yet another variant of the lifting idea. One 
 of the most general versions can be found in \cite[Lemma 4.1]{cossidente2019combining}:
\begin{Lemma}
  \label{lemma_construction_1}
  For a subspace distance $d$, 
  let $\bar{n}=\left(n_1,\dots,n_l\right)\in\mathbb{N}^l$, where $l \ge 2$, be such that $\sum_{i=1}^l n_i=n$ and $n_i\ge k$ for all $1\le i\le l$. 
  Let $\mathcal{C}_i$ be an $(n_i,\star,d,k)_q$ {\cdc} and $\mathcal{M}_i$ be a $(k\times n_i,\tfrac{d}{2})_q$-rank metric code for $1\le i\le l$. Then 
  $\mathcal{C}=\bigcup_{i=1}^l\mathcal{C}^i$, where
  \begin{eqnarray*}
    \mathcal{C}^i=\Big\{\left\langle M_1|\dots|M_{i-1}|E(U_i)|M_{i+1}|\dots|M_l\right\rangle &\!\!\!\!\!\!\!:\!\!\!\!\!\!\!&
    U_i\in\mathcal{C}_i, M_j\in\mathcal{M}_j, \,\forall 1\le j\le l, i\neq j,\\ 
    &&\text{ and } \operatorname{rk}(M_j)\le k-\tfrac{d}{2}, \,\forall 1\le j<i \Big\},
  \end{eqnarray*} 
  is an $(n,\star,d,k)_q$ {\cdc} of cardinality
  $$
    \#\mathcal{C}=\sum_{i=1}^l \left(\prod_{j=1}^{i-1} \#\left\{M\in\mathcal{M}_j\,:\, \rk(M)\le k-\tfrac{d}{2}\right\}\right)\cdot 
    \#\mathcal{C}_i\cdot\left(\prod_{j=i+1}^{l} \#\mathcal{M}_j\right).
  $$
\end{Lemma}
So, if we assume that the $\cM_j$ are additive {\MRD} codes, then using Equation~(\ref{eq_rank_distribution_additive_MRD_code}) directly gives:
\begin{Corollary}(\cite[Corollary 4.2]{cossidente2019combining})
  \label{cor_construction_1}
  Let $d$ be a subspace distance, $\bar{n}=\left(n_1,\dots,n_l\right)\in\mathbb{N}^l$, and $l\ge 2$, be such that $\sum_{i=1}^l n_i=n$ and $n_i\ge k$ for all $1\le i\le l$. 
  Then, we have $A_q(n,d;k)\ge$
  $$
    \sum_{i=1}^l \left(\prod_{j=1}^{i-1} \left(1+\sum_{r=\tfrac{d}{2}}^{k-\tfrac{d}{2}}a(q,k,n_j,\tfrac{d}{2},r)\right)\right)\cdot 
    A_q(n_i,d;k)\cdot\left(\prod_{j=i+1}^{l}m(q,k,n_j,\tfrac{d}{2})\right).
  $$
\end{Corollary}
Of course, Lemma~\ref{lemma_construction_1} can also be applied if $\cM_j$ is not additive or not an {\MRD} code. As an example we consider the 
$(3\times 4,3)_2$ {\MRD} codes classified in \cite{hkk_partial_plane_spreads}. Up to isomorphism there are $7$ linear and $30$ non-linear such codes. Considering a 
coset, i.e.\ adding an arbitrary matrix in $\F_2^{3\times 4}$ to all codewords, does not change the minimum rank distance but eventually the rank distribution. Here 
the occurring rank distributions are given by $0^1 3^{15}$, $2^7 3^9$, and $1^1 2^4 3^{11}$. 
Rank-metric codes of constant rank with a lower bound on the minimum rank-distance have been studied in \cite{gadouleau2010constant} and generalized in \cite{DH,liu2019parallel}. 
As rank metric codes with a given minimum rank distance and an upper bound on the occurring ranks pop up here, we propose the study of their sizes as an interesting open 
research problem. Improvements for these rank metric codes can directly result in improved constructions for {\cdc}s.

\subsection{Adding additional codewords to {\cdc}s constructed via a skeleton code}
\label{subsec_additional_codewords}
In Subsection~\ref{subsec_skeleton_codes} we have considered the construction of a {\cdc} $\cC$ as a union of subcodes $\cC_{\cV_i}$ via a generalized skeleton code 
$\cS=\left\{\cV_1,\dots,\cV_s\right\}$, see Theorem~\ref{thm_generalized_skeleton_code} for the details. Constructions for the subcodes $\cC_{\cV_i}$ were the topic of 
Subsection~\ref{subsec_rank_metric}. For special choices of the (generalized) skeleton code $\cS$ there is additional structure that allows the addition of further codewords.
For an ordinary skeleton code, with nodes corresponding to a single pivot vector as occurring in the Echelon--Ferrers construction, one can observe that the removal of some 
specific dots from a given Ferrers diagram does not decrease the upper bound on the code size from Theorem~\ref{thm_upper_bound_ef}. Those dots are called 
\emph{pending dots} and their positions can be used to construct additional codewords \cite{trautmann2010new}. Ferrers diagrams can also contain several pending dots, which may be 
pooled to a so-called \emph{pending block} allowing more sophisticated additions of codewords, see \cite{silberstein2015error} for the details.

Here we want to focus on {\cdc}s $\cC=\cup_{i=1}^l \cC^i$ according to Lemma~\ref{lemma_construction_1}, where we have the following structural result.
\begin{Lemma}(\cite[Lemma 4.3]{cossidente2019combining})
  \label{lemma_special_substructure} 
With the same notation used in Lemma~\ref{lemma_construction_1}, set $\sigma_i=\sum_{j=1}^{i} n_j$, $1 \le i \le l$ and $\sigma_0 = 0$. Let $E_i$ denote the $(n-n_i)$-subspace of 
$\F_q^n$ consisting of all vectors in $\F_q^n$ that have zeroes for the coordinates between $\sigma_{i-1}+1$ and $\sigma_i$ for all $1\le i\le l$. Then, the elements of 
$\mathcal{C}^i$ are disjoint from $E_i$ for all $1\le i\le l$. 
\end{Lemma} 
Similar as for the Hamming metric we write $\ds(\cC,\cC'):=\min\{\ds(U,U')\,:\,U\in\cC,U'\in\cC'\}$.
\begin{Lemma}(\cite[Lemma 4.4]{cossidente2019combining})
  \label{lemma_construction_2}
Let $\mathcal{C}$ be a subspace code as in Lemma~\ref{lemma_construction_1} with corresponding $\bar{n}\in\mathbb{N}^l$, $\bar{a}=\left(a_1,\dots,a_l\right)\in \mathbb{N}^l$ 
and $\bar{b}=\left(b_1,\dots,b_l\right)\in\mathbb{N}^l$ with $\sum_{i=1}^l a_i=k$, $\sum_{i=1}^l b_i=k-\tfrac{d}{2}$, and $\tfrac{d}{2}\le a_i, 
b_i<a_i\le n_i$, for all $1\le i \le l$. For an integer $r$, let $\mathcal{D}_i^j$ be $(n_i, \star, d, a_i)_q$ {\cdc}s, for all $1\le i\le l$ and all $1\le j\le r$, such that 
$\ds(\mathcal{D}_i^{j_1},\mathcal{D}_i^{j_2})\ge 2a_i-2b_i$, for all $1\le i\le l$ and all $1\le j_1<j_2\le r$. Then, there exists an $(n, \star, d, k)_q$ {\cdc}, say $\mathcal{D}$, 
with cardinality
$$
    \#\mathcal{D}=\sum_{j=1}^r \prod_{i=1}^l \#\mathcal{D}_i^j ,
$$
such that $\mathcal{C} \cap \mathcal{D}=\emptyset$ and $\mathcal{C}\cup\mathcal{D}$ is also an $(n, \star, d, k)_q$ {\cdc}.
\end{Lemma}
We remark that we can also take different subcodes as in Lemma~\ref{lemma_construction_2} and combine these codes exploiting the underlying pivot structure. To this end let $\mathcal{D}$ be 
the code for $\bar{a}=(a_1,\dots,a_l)$ and $\mathcal{D}'$ be the code for $\bar{a}'=(a'_1,\dots,a'_l)$ according to Lemma~\ref{lemma_construction_2}. (The corresponding vectors 
$\bar{b}$ and $\bar{b}'$ are not relevant for the subsequent analysis.) From Lemma~\ref{lemma_dist_subspace_hamming} we conclude 
\begin{equation}
  \label{eq_a_bar_combination}
  \ds(\mathcal{D},\mathcal{D}')\ge \sum_{i=1}^l \left|a_i-a'_i\right|
\end{equation}
and refer to \cite{cossidente2019combining} for an example. So, in general 
we will consider a {\cdc} given by
\begin{equation}
  \cC=\cup_{i=1}^s \cC^i \,\cup\, \cup_{j=1}^{t} \cD^j, 
\end{equation}   
where $s=2$ (and $t$ is rather small) in most applications. The compatibility of the subcodes $\cC^i$ and $\cD^j$ is described in terms of the Hamming distance. For the 
(known) construction of the subcodes $\cC^i$ and $\cD^j$ itself, rank metric codes play a major role. With respect to constructions for the $\cD^j$ according to 
Lemma~\ref{lemma_construction_2} we remark that for each $1\le i\le l$, the {\cdc} $\bigcup_{j=1}^r \mathcal{D}_i^j$ is an $(n_i,\star,2a_i-2b_i,a_i)_q$ {\cdc}.  
Partitioning it into subcodes with subspace distance $d>2a_i-2b_i$ is a hard problem in general and was e.g.\ considered in the context of the \emph{coset construction} 
for {\cdc}s, see \cite{heinlein2017coset}. We have a closer look at this problem in Subsection~\ref{subsec_improved_packings}. Restricting to lifted {\MRD} codes an analytic 
construction, using rank metric codes, was given in \cite[Corollary 4.5]{cossidente2019combining}:
\begin{Corollary}
  \label{cor_construction_2}
  In Lemma~\ref{lemma_construction_2} one can achieve
  $$
    \#\mathcal{D}\ge \min\{\alpha_i\,:\,1\le i\le l\}\cdot\prod_{i=1}^l m\!\left(q,a_i,n_i-a_i,\tfrac{d}{2}\right),
  $$
  where $\alpha_i=m\!\left(q,a_i,n_i-a_i,a_i-b_i\right)/m\!\left(q,a_i,n_i-a_i,\tfrac{d}{2}\right)$.
\end{Corollary}
The work in \cite[Section 4]{cossidente2019combining} initiated many improved constructions for {\cdc}s. Several of them started from 
Lemma~\ref{lemma_special_substructure} and improved Lemma~\ref{lemma_construction_2} and Corollary~\ref{cor_construction_2}, 
see e.g.\ \cite{he2020construction,he2021new,lao2020parameter,liu2019parallel,niu2020new}. We will briefly discuss this possibility in Subsection~\ref{subsec_exploit}. 

\subsection{Special constructions for {\cdc}s}
\label{subsec_special_constructions}
For a few parameters special constructions for {\cdc}s have been presented in the literature. Since we use some of them in improved constructions for other parameters 
as subcodes, we here summarize the necessary details.

\begin{Proposition}(\cite{TableSubspacecodes,heinlein2019subspace,honold2016putative})
  \label{prop_7_4_3}
  $A_2(7,4; 3)\ge 333$, $A_3(7,4; 3)\ge 6978$, and $A_q(7,4; 3) \ge q^8 + q^5 + q^4 + q^2 -q$ for $q\ge 2$.
\end{Proposition}
  
\begin{Proposition}(\cite{braun2018new,cossidente2019combining,cossidente2016subspace,heinlein2017coset})  
  \label{prop_8_4_4}
  $A_2(8,4;4)\ge 4801$ and $A_q(8,4;4)\ge q^{12} + q^2(q^2 + 1)^2(q^2 + q + 1) + 1$ for $q\ge 2$.
\end{Proposition}

Other examples with small parameters, that are not used in our examples of improved constructions but are very likely to be contained 
in similar constructions are: 

\begin{Proposition}(\cite{hkk77})
  \label{prop_6_4_3}
  $A_2(6,4;3)=77$ and $A_q(6,4; 3) \ge q^6 + 2q^2 + 2q + 1$ for $q\ge 2$.  
\end{Proposition}
     
\begin{Proposition}(\cite{braun2018new})  
  \label{prop_ilp_lb_binary}
  $A_2(8,4;3)\ge 1326$,
  $A_2(9,4;3)\ge 5986$,
  $A_2(10,4;3)\ge 23870$, and 
  $A_2(11,4;3)\ge 97526$.
\end{Proposition}
     
In \cite[Section 5]{cossidente2019combining} another general construction strategy for constant dimension codes, outside of the here presented scheme, is considered. 
As an example we mention:
\begin{Proposition}(\cite{kurz2020subspaces})
  \label{prop_9_4_3}
  $A_q(9, 4; 3) \ge  q^{12} + 2q^8 + 2q^7 + q^6 + 2q^5 + 2q^4 -2q^2 -2q + 1$ for $q\ge 2$.
\end{Proposition}  

\begin{table}[htp!]
  \begin{center}
    \begin{tabular}{ll}
      \hline
      {\cdc} & constant dimension code \\
      {\MRD} & maximum rank distance\\ 
      {\FDRM} & Ferrers diagram rank metric\\ 
      $\ds(U,W)$ & subspace distance between codewords $U$ and $W$ \\ 
      $\ds(\cC)$ & minimum subspace distance of a {\cdc} $\cC$\\
      $\dham(u,w)$ & Hamming distance between codewords $u$ and $w$\\
      $\dham(\cS)$ & minimum Hamming distance of $\cS$\\ 
      $\dham(\cV,\cV')$ & minimum Hamming distance between a codeword in $\cV$ and\\ 
         & a codeword in $\cV'$\\
      $\dr(A,B)$ & rank distance between two matrices $A$ and $B$\\ 
      $\rk(A)$ & rank of a matrix $A$\\ 
      $\cG_1(n,k)$ & set of binary vectors of length $n$ and Hamming weight $k$\\
      $\cG_q(n,k)$ & set of $k$-dimensional subspaces in $\F_q^n$\\
      $\gaussmnum{n}{k}{q}$ & Gaussian binomial coefficient; $\# \cG_q(n,k)$\\
      $A_q(n,d;k)$ & maximum possible cardinality of a {\cdc} $\cC\subseteq\cG_{q}(n,k)$ \\
      & with minimum subspace distance at least $d$\\      
      $m(q,m,n,\dr)$ & number of codewords of an $(m\times n,\dr)_q$-{\MRD} code\\
      $a(q,m,n,\dr,r)$ & number of codewords of rank $r$ in an additive $(m\times n,\dr)_q$-{\MRD} code\\ 
      $E(U)$, $E(M)$ & matrix $M$ or generator matrix of $U$ in reduced row echelon form\\
      $v(U)$, $v(M)$ & pivot vector\\ 
      ${n_1 \choose k_1},\dots, {n_l \choose k_l}$ & set of binary vectors\\
      $T(U)$ & Ferrers tableaux\\
      $\cF(U)$, $\cF(v)$ & Ferrers diagram\\ 
      $A_q(n,d;k;\cV)$ & max.\ possible cardinality of a {\cdc} $\cC\subseteq\cG_{q}(n,k)$ with min.\ \\ 
      & subspace distance at least $d$ whose codewords have pivot vectors in $\cV$\\
      $I_k$ & $k\times k$ unit matrix\\ 
      $e_i$ & unit vector with a one at position $i$\\ 
      \hline
    \end{tabular}
    \caption{Notation and abbreviations.}
    \label{table_notation}  
  \end{center}
\end{table}

\section{Improved constructions}
\label{sec_improved_constructions}

The aim of this section is to highlight the general potential for improved constructions for constant dimension codes based on general construction strategies presented in 
the literature. We structure the different lines of attack into several subsections. In this context we would like to point to the discussion on 
rank metric codes with restricted ranks at the end of Subsection~\ref{subsec_rank_metric}.

\subsection{New generalized skeleton codes}
\label{subsec_new_skeleton_codes}

Computing good skeleton codes is a hard combinatorial problem. For recent improvements for the Echelon-Ferrers construction we e.g.\ refer to  \cite{feng2020bounds}. 
In the context of the linkage construction similar improvements can be e.g.\ found in \cite{he2020improving,kurz2020lifted}. Taking codes from Subsection~\ref{subsec_special_constructions} 
as subcodes, only knowing their attained pivot vectors or a superset thereof, as subcodes, can also lead to (tiny) improvements. 

\begin{Proposition}
  \label{prop_ex_generalized_skeleton_code}
  $A_2(11,\!4;4)\ge 2383085$, $A_3(11,\!4;4)\ge 10639658703$, and $A_q(11,\!4;4)\ge q^{21}+q^{17}+2q^{15}+3q^{14}+4q^{13}+q^{12}+q^{11}+q^9+q^8+2q^7+2q^6+2q^5+q^4+q^2-q$ for $q\ge 2$.
\end{Proposition}
\begin{proof}
  We choose a generalized skeleton code $\cS$ with vertices {\tiny $\left(\!{4\choose 0},\!{7\choose 4}\!\right)$, $0 0 0 1 0 0 0 0 1 1 1$, $0 0 0 1 0 1 0 0 0 1 1$,  
    $0 0 0 1 1 0 0 0 0 1 1$,\, 
    $0 0 0 1 1 0 0 0 1 1 0$,\,
    $0 0 1 0 0 0 0 1 0 1 1$,\, 
    $0 0 1 0 0 0 0 1 1 0 1$,\, 
    $0 0 1 0 0 0 0 1 1 1 0$,\, 
    $0 0 1 0 0 1 0 0 1 0 1$,\, 
    $0 0 1 0 0 1 0 0 1 1 0$,\, 
    $0 0 1 0 0 1 0 1 0 0 1$, 
    $0 0 1 0 1 0 0 0 1 0 1$, 
    $0 0 1 1 0 0 0 0 1 1 0$, 
    $0 0 1 1 0 1 0 1 0 0 0$, 
    $0 1 1 0 0 0 1 0 0 0 1$, 
    $1 0 0 0 0 1 0 1 1 0 0$, 
    $1 0 0 0 1 0 0 1 0 0 1$, 
    $1 0 0 1 1 1 0 0 0 0 0$ 
    $1 0 1 0 0 0 0 0 0 1 1$, and 
    $1 0 1 0 0 1 1 0 0 0 0$}, so that
    $$
      A_q(11,4;4)\ge q^{21}+q^{17}+2q^{15}+3q^{14}+4q^{13}+q^{12}+q^{11}+q^9+2q^7+2q^6+q^5+A_q(7,4;4).
    $$
    Using $A_q(7,4;4)=A_q(7,4;3)$ and Proposition~\ref{prop_7_4_3} gives the stated results. 
\end{proof}
We remark that the previously best known lower bound was given by the Echelon-Ferrers construction yielding e.g.\ $A_2(11,4;4)\ge 2383041$ for $q=2$.

While listing $19$ explicit pivot vectors as elements of a generalized skeleton $\cS$ is still manageable, we need a more more compact representation for larger 
instances. To this end we replace each vector $v\in\F_2^n$ by the integer $\sum_{i=1}^n v_i\cdot 2^{n-i}$. As an example, the integer $24672$ corresponds to the vector
$1 1 0 0 0 0 0 0 1 1 0 0 0 0 0\in\F_2^{15}$. Starting from an integer, the value of $n$ needs to be clear from the context. In our next example we 
show that generalized skeleton codes with two vertices corresponding to more than one pivot vector can also lead to improved constructions.

\begin{Proposition}
  \label{prop_ex_generalized_skeleton_code2}
  $A_2(15,4;4)\ge 10073483885$ and $A_q(15,4;4)\ge q^{33} + q^{29} + q^{28} + 3q^{27} + 2q^{26} + 3q^{25} + q^{24} + q^{23} + 2q^{21} + 2q^{19} + 3q^{18} 
  + 5q^{17} + q^{16} + 4q^{15} + 6q^{14} + 11q^{13} + 10q^{12} + 13q^{11} + 11q^{10} + 8q^9 + 4q^8 + 3q^7 + 2q^6 + 2q^5 + q^4 + q^2 - q$ for $q\ge 2$. 
\end{Proposition}  
\begin{proof}
  We choose a generalized skeleton code $\cS$ with vertices {\tiny $\left({8\choose 4},{7\choose 0}\right)$, $\left({8\choose 0},{7\choose 4}\right)$, 
 $24672$, $6240$, $12368$, $18512$, $20528$, $20552$, $1632$, $10288$, $10312$, $12328$, $24600$, $18472$, $480$, $848$, $3140$, $6168$, $1232$, $1328$, 
 $1352$, $4676$, $5156$, $5186$, $688$, $712$, $808$, $1560$, $2596$, $2626$, $3106$, $8516$, $9236$, $9281$, $24582$, $1192$, $4642$, $16580$, $16676$, 
 $16706$, $16916$, $16961$, $17420$, $17426$, $17441$, $408$, $2324$, $2369$, $3089$, $6150$, $8356$, $8386$, $8482$, $8716$, $8722$, $8737$, $9226$, 
 $12293$, $4244$, $4289$, $4364$, $4370$, $4385$, $4625$, $5129$, $16546$, $16906$, $18437$, $20483$, $1542$, $2188$, $2194$, $2209$, $2314$, $2569$, 
 $8465$, $10243$, $4234$, $16529$, $16649$, $390$, $773$, $8329$, $1157$, $1283$, and $643$, }, so that Inequality~(\ref{ie_EF_gen}) and (\ref{ie_lifted}) give
  $A_q(15,4;4)\ge 18727097+A_q(8,4;4)\cdot q^{21}+q^{21} +2q^{19} +3q^{18} +5q^{17} +q^{16} +4q^{15} +6q^{14} +11q^{13} +10q^{12} +13q^{11} +11q^{10}
  +8q^9 +3q^8 +3q^7 +2q^6 +q^5+A_q(7,4;4)$. Using $A_q(7,4;4)=A_q(7,4;3)$, Proposition~\ref{prop_7_4_3}, and Proposition~\ref{prop_8_4_4} gives the stated result.
\end{proof}
We remark that the previously best known lower bound was given in \cite{kurz2020lifted} with e.g.\ $A_2(15,4;4)\ge 10073483841$ for $q=2$.

\subsection{Improved packings}
\label{subsec_improved_packings}
Our next starting point for improved constructions is Lemma~\ref{lemma_construction_2}. As an example we consider the parameters $l=2$, $n_1=5$, $n_2=5$, $a_1=2$, 
$a_2=3$, $b_1=1$, and $b_2=2$, i.e., we are aiming at a lower bound for $A_q(10,4;5)$. Lemma~\ref{lemma_construction_1} and Corollary~\ref{cor_construction_1} give a 
$(10,\star,4,5)_q$ {\cdc} $\cC$ with
\begin{equation}
  \#\cC=q^{20}+\gaussmnum{5}{2}{q} \cdot \left(q^{10} - q^7 - q^6 + q^2 +q-1\right)+1,
\end{equation}
i.e., $\#\cC=1178312$ for $q=2$. For our specific choice $\bar{n}=(n_1,n_2)=(5,5)$, $\bar{a}=(a_1,a_2)=(2,3)$, and $\bar{b}=(b_1,b_2)=(1,2)$ Corollary~\ref{cor_construction_2} 
gives a $(10,\star,4,5)_q$ {\cdc} $\cD$ such that $\cC\cap\cD=\emptyset$ and $\ds(\cC\cup\cD)\ge 4$, where
$\#\cD\ge q^9$, i.e., $\#\cD\ge 512$ for $q=2$. Going back to Lemma~\ref{lemma_construction_2} the actual conditions are that the $\cD_1^j$ are $(5,\star,4,2)_q$ {\cdc}s 
for all $1\le j\le r$ with $\ds(\cD_1^j,\cD_1^{j'})\ge 2$ for all $1\le j<j'\le r$ and that the $\cD_2^j$ are $(5,\star,4,3)_q$ {\cdc}s 
for all $1\le j\le r$ with $\ds(\cD_2^j,\cD_2^{j'})\ge 2$ for all $1\le j<j'\le r$. Setting $\cD_2^j=\left(\cD_1^j\right)^\perp$ it suffices to give a construction for the $\cD_1^j$. 
The condition $\ds(\cD_1^j,\cD_1^{j'})\ge 2$ for all $1\le j<j'\le r$ just says that we can pack each of the $\gaussmnum{5}{2}{q}=q^6 + q^5 + 2q^4 + 2q^3 + 2q^2 + q + 1$ 
$2$-dimensional subspaces of $\F_q^5$ into at most one $\cD_1^j$. So, let $\cL$ be the set of all $\gaussmnum{5}{2}{q}$ $2$-dimensional subspaces of $\F_q^5$ and $j=1$. 
Now we iteratively and greedily select some large $(5,\star,4;2)_q$-subcode $\cD_1^j$ from $\cL$, remove the codewords from $\cD_1^j$ from $\cL$, and increase $j$ by $1$ until 
$\cL$ is empty. As a result we obtain $14$ codes with $\#\cD_1^j=9$ and one code $\cD_1^j$ for each cardinality in $\{1,2,5,6,7,8\}$. Note that $14\cdot 9+8+7+6+5+2+1=155$ and 
$14\cdot 9^2 + 8^2 + 7^2 + 6^2 + 5^2 + 2^2 + 1^2=1313$, so that $A_2(10,4;5)\ge 1178312+1313=1179625$. Since $A_2(5,4;2)=9$ we have $\#\cD_1^j\le 9$ and 
$\left\lfloor \gaussmnum{5}{2}{2}/9\right\rfloor=17$ implies that at most $17$ $\cD_1^j$ can have the maximum cardinality $9$. From $155-17\cdot 9=2$ we conclude 
$\sum_{j=1}^{r} \left(\#\cD_1^j\right)^2\le 17\cdot 9^2+2^2=1381$.   

\begin{Definition}
  \label{def_E_a}
  Let $l\ge 2$, $d\ge 2$ with $d\equiv 0\pmod 2$, $\bar{n}=\left(n_1,\dots,n_l\right)\in\N^l$, $n:=\sum_{i=1}^l n_i$, $\bar{a}=\left(a_1,\dots,a_l\right)$ with $a_i\ge d/2$ for all 
  $1\le i\le l$, and $k=\sum_{i=1}^l a_i$. Let $F_i$ denote the subspace spanned by the unit vectors $e_h$ for $\sum_{j=1}^{i-1} n_j<h\le \sum_{j=1}^{i} n_j$, where $1\le i\le l$.    
  By $E_q(\bar{n},\bar{a},d)$ we denote denote the maximum cardinality $M$ of an $(n,M,d,k)_q$ {\cdc} $\cD$ such that every codeword $U\in\cD$ satisfies 
  $\dim(U\cap F_i)=a_i$ for $1\le i\le l$.  
\end{Definition}

So, we e.g.\  have $E_q\big((5,5),(2,3),4\big)\ge q^9$ and $E_2\big((5,5),(2,3),4\big)\ge 1313$. The general construction strategy in our situation can be described as
\begin{eqnarray}
  A_q(10,4;5)&\ge& A_q\!\left(10,4;5;{5\choose 5},{5\choose 0}\right)+A_q\!\left(10,4;5;{5\choose \le 2},{5\choose \ge 3}\right)\notag\\ &&+E_q\big((5,5),(2,3),4\big)\label{ie_10_4_5_construction}. 
\end{eqnarray}
The advantage of such a description is that the three parts can be considered separately.

In order to improve upon Corollary~\ref{cor_construction_2} in general we have to introduce a bit more notation and state the key observation of its proof.

\begin{Lemma}(Cf.~\cite[Lemma 2.5]{lao2020parameter} and the proof of \cite[Corollary 4.5]{cossidente2019combining})
  \label{lemma_cosets}
  Let $\cF$ be a Ferrers diagram and $\cM$ be a corresponding linear {\FDRM} code with minimum rank distance $\delta$. If $\cM$ is a subcode of a 
  linear {\FDRM} code $\cM'$ with minimum rank distance $\delta'<\delta$ and Ferrers diagram $\cF$, then there exist 
  {\FDRM} codes $\cM_i$ with Ferrers diagram $\cF$ for $1\le i\le s:=\#\cM'/\#\cM$ satisfying 
  \begin{enumerate}
    \item[(1)] $\dr(\cM_i)\ge \delta$ for all $1\le i\le s$;
    \item[(2)] $\dr(\cM_i,\cM_j)\ge \delta'$ for all $1\le i<j\le s$; and
    \item[(3)] $\cM_1,\dots,\cM_r$ is a partition of $\cM'$.
  \end{enumerate} 
\end{Lemma}
\begin{proof}
  For each $M'\in \cM'$ the code $\cM+M':=\{M+M'\,:\, M\in \cM\}$ is $\FDRM$ with Ferrers diagram $\cF$ and minimum rank distance $\delta$. For 
  $M',M''\in \cM'$ we have $M'+\cM=M''+\cM$ iff $M'-M''\in \cM$ and $M'+\cM\cap M''+\cM=\emptyset$ otherwise. Now let $\cM_1,\dots,\cM_s$ be the 
  $s=\#\cM'/\#\cM$ different codes $M+\cM$, which are cosets of $\cM$ in $\cM'$ and partition $\cM'$. Since all elements of $\cM_i$ and $\cM_j$ 
  are different elements of $\cM'$ we have $\dr(\cM_i,\cM_j)\ge \delta'$ for all $1\le i<j\le s$. 
\end{proof}
Choosing $\cF$ as $a\times b$ rectangular Ferrers diagram, we end up with \cite[Lemma 2.5]{lao2020parameter}. In the proof of 
\cite[Corollary 4.5]{cossidente2019combining} this lemma is indirectly applied with $a=a_i$ and $b=n_i-a_i$. By $m(q,\cF,\dr)$ we denote 
the maximum cardinality of an {\FDRM} code with Ferrers diagram $\cF$ and minimum rank distance $\dr$. This generalizes the notion 
of $m(q,m,n,\dr)$ for the cardinality of {\MRD} codes choosing $\cF$ as $m\times n$ rectangular Ferrers diagram. Note that for minimum rank 
distance $\delta=2$ the upper bound from \cite[Theorem 1]{etzion2009error}, cf.\ Theorem~\ref{thm_upper_bound_ef}, can always be attained by 
linear rank metric codes. Moreover, the only choice for $\delta'$ then is $\delta'=1$ and $\cM'$ consists of all matrices with Ferrers diagram $\cF$. 
Thus, $\cM'$ is automatically linear and contains $\cM$ as a subcode. 

Now we are ready to describe the link to Lemma~\ref{lemma_construction_2}. We write $\cF(v)$ for a Ferrers diagram whose pivot vector is given by $v$.  
Let $\cF$ be a Ferrers diagram with a pivot vector contained in $\cG_1(n_i,a_i)$. We apply Lemma~\ref{lemma_cosets} for $\delta=2$ and $\delta'=1$. With the corresponding $\cM_j$ for
$1\le j\le r:=m(q,\cF,1)/m(q,\cF,2)$ we can set
\begin{equation}
  \cD_i^j=\left\{\left\langle I_{a_i}|M\right\rangle\,:\,M\in \cM_j\right\}
\end{equation}
for $1\le j\le r$. For the sake of simplicity, let us restrict to the parameters $l=2$, $n_1=n_2$, and $a_1=a_2$. By choosing 
\begin{equation}
  \cD=\cup_{j=1}^r \left\{U\times U'\,:\, U\in \cD_1^j,U'\in \cD_2^j\right\}
\end{equation}
we obtain a code $\cD$ of cardinality $m(q,\cF,1)\cdot m(q,\cF,2)$ that goes in line with the conditions of Lemma~\ref{lemma_construction_2}. 
Choosing $\cF$ as a rectangular Ferrers diagram of maximum shape gives Corollary~\ref{cor_construction_2}. However, for minimum subspace distance $d=4$ we 
can choose the union of these codes for all possible Ferrers diagrams:
\begin{Proposition}
  \label{prop_all_ferrers_diagrams}
  $$
    E_q((n',n'),(a',a'),4)\ge \sum_{v\in \cG_1(n',a')} m(q,\cF(v),1)\cdot m(q,\cF(v),2).
  $$  
\end{Proposition}

\begin{table}[htp]
\begin{center}
  \begin{tabular}{lll}
    \hline
    pivot vector & size $m(q,\cF,2)$ & $\#$ of cosets $m(q,\cF,1)/m(q,\cF,2)$ \\  
    \hline
    $11000$ & $q^3$ & $q^3$ \\ 
    $10100$ & $q^2$ & $q^3$ \\ 
    $10010$ & $q$   & $q^3$ \\
    $10001$ & $1$   & $q^3$ \\ 
    $01100$ & $q^2$ & $q^2$ \\ 
    $01010$ & $q$   & $q^2$ \\
    $01001$ & $1$   & $q^2$ \\
    $00110$ & $1$   & $q^2$ \\
    $00101$ & $1$   & $q$   \\
    $00011$ & $1$   & $1$   \\  
    \hline
  \end{tabular}
  \caption{Data for Lemma~\ref{lemma_cosets} with $\cF\in\cG_1(5,2)$.}
  \label{table_cosets_5_2}
\end{center}      
\end{table}

For $n'=5$ and $a'=2$ we obtain, see Table~\ref{table_cosets_5_2} for the details, 
\begin{equation}
  E_q((5,5),(2,2),4)\ge q^9+q^7+q^6+q^5+q^4+q^3+2q^2+q+1,
\end{equation}
so that e.g.\ $E_2((5,5),(2,2),4)\ge 771$.

If we choose $\cD_2^j=\left(\cD_1^j\right)^\perp$, as done at the beginning of this subsection, we obtain:
\begin{Proposition}
  $$
    E_q((n',n'),(a',n'-a'),4)\ge \sum_{v\in \cG_1(n',a')} m(q,\cF(v),1)\cdot m(q,\cF(v),2).
  $$  
\end{Proposition}
For our specific parameters we obtain 
\begin{equation}
  E_q((5,5),(2,3),4)\ge q^9+q^7+q^6+q^5+q^4+q^3+2q^2+q+1,
\end{equation}
so that e.g.\ $E_2((5,5),(2,3),4)\ge 771$.

If we choose $\cD_2^j=\left(\cD_1^j\right)^\perp$, as done at the beginning of this subsection, we obtain:
\begin{Proposition}
  $$
    E_q((n',n'),(a',n'-a'),4)\ge \sum_{\cF\in \cG_1(n',a')} m(q,\cF,1)\cdot m(q,\cF,2)
  $$  
\end{Proposition}
For our specific parameters we obtain 
\begin{equation}
  E_q((5,5),(2,3),4)\ge q^9+q^7+q^6+q^5+q^4+q^3+2q^2+q+1,
\end{equation}
so that e.g.\ $E_2((5,5),(2,3),4)\ge 771$.

Let us consider the initial packing or partitioning problem again, i.e., pack or partition the $\gaussmnum{5}{2}{q}$ $2$-dimensional subspaces of $\F_q^5$ 
into {\cdc}s $\cD_1^j$ with $\ds(\cD_1^j)\ge 4$. In Proposition~\ref{prop_all_ferrers_diagrams} and Table~\ref{table_cosets_5_2} the 
$\cD_1^j$ all have the same pivot vector. Combining codewords with pivot vector $11000$ with those with pivot vector $00110$ allows us to choose 
$\#\cD_1^j=q^3+1$. However, we can choose only $\min\!\left\{q^3,q^2\right\}=q^2$ translates, i.e., different corresponding indices $j$. This leaves $q^3-q^2$ 
translates for the pivot vector $11000$. Using the packing scheme from Table~\ref{table_cosets_5_2_packed} we obtain:  
\begin{Proposition}
  \label{prop_eq_5_5_2_2}
  $$
    E_q((5,5),(2,2),4),E_q((5,5),(2,3),4)\ge q^9 + q^7 + q^6 + 7q^5 + 5q^4 + 3q^3 + 2q^2 + q + 1
  $$
\end{Proposition}
For $q=2$ we obtain $E_2((5,5),(2,2),4),E_2((5,5),(2,3),4)\ge 1043$. Since $1043$ is much smaller than $1313$, there still seems to be a lot of space for improvements 
for general field sizes $q$. 

\begin{table}[htp]
\begin{center}
  \begin{tabular}{lll}
    \hline
    skeleton code & size & $\#$ of used cosets \\ 
    \hline
    $\{11000,00110\}$ & $q^3+1$ & $q^2$ \\
    $\{11000,00101\}$ & $q^3+1$ & $q$ \\
    $\{11000,00011\}$ & $q^3+1$ & $1$ \\
    $\{11000\}$       & $q^3$ & $q^3-q^2-q-1$ \\      
    \hline
    $\{10100,01010\}$ & $q^2+q$ & $q^2$ \\
    $\{10100,01001\}$ & $q^2+1$ & $q^2$ \\
    $\{10100\}$       & $q^2$   & $q^3-2q^2$ \\
    \hline
    $\{01100,10010\}$ & $q^2+q$   & $q^2$ \\
    $\{10010\}$       & $q$       & $q^3-q^2$ \\
    $\{10001\}$       & $1$       & $q^3$ \\
    \hline   
  \end{tabular}
  \caption{Packing scheme for Proposition~\ref{prop_eq_5_5_2_2}.}
  \label{table_cosets_5_2_packed}
\end{center}      
\end{table}

Combining Inequality~(\ref{ie_10_4_5_construction}) with Proposition~\ref{prop_eq_5_5_2_2} gives:
\begin{Corollary}
  \begin{eqnarray*} A_q(10,4;5) &\ge& q^{20}+\gaussmnum{5}{2}{q} \cdot \left(q^{10} - q^7 - q^6 + q^2 +q-1\right)+1 \\ 
  && + q^9 + q^7 + q^6 + 7q^5 + 5q^4 + 3q^3 + 2q^2 + q + 1
  \end{eqnarray*}
\end{Corollary}     

\bigskip

Let us consider an improved construction for $A_q(12,6;6)$ as a second example. Here the desired minimum subspace distance is strictly larger than $4$, so that 
we cannot apply Proposition~\ref{prop_all_ferrers_diagrams}. However, we again end up with some kind of packing problem where we can state a slightly improved construction 
being parametric in the field size $q$. We choose $l=2$, $\bar{n}=(6,6)$ in Lemma~\ref{lemma_construction_1} and Corollary~\ref{cor_construction_1}. Taking 
Lemma~\ref{lemma_special_substructure} and Lemma~\ref{lemma_construction_2} into account we have  
\begin{eqnarray}
  A_q(12,6;6)&\ge& A_q\!\left(12,6;6;{6\choose 6},{6\choose 0}\right)+A_q\!\left(12,6;6;{6\choose \le 3},{6\choose \ge 3}\right)\notag\\ &&+E_q\big((6,6),(3,3),6\big)\label{ie_12_6_6_construction}. 
\end{eqnarray}
We remark that the previously best known lower bound for $A_q(12,6;6)$, described in \cite{cossidente2019combining}, indirectly gives 
$E_q((6,6),(3,3),6)\ge q^9+2q^3$. The corresponding packing problem is the following. Let $\cB$ be an $(6,\star,4,3)$ {\cdc} that is partitioned into $(6,\star,6,3)$ {\cdc}s 
$\cB^j$ for $1\le j\le r$, where $r\ge 1$ is a suitable integer. Then, by choosing $\cD_1^j=\cB^j$ and $\cD_2^j=\cB^j$ Lemma~\ref{lemma_construction_2} gives
$$
  E_q((6,6),(3,3),6)\ge \sum_{j=1}^r \left(\#\cB^j\right)^2.
$$
The pivot vector $111000$ gives codes of size $q^3$ in $q^3$ different cosets and the pivot vector $000111$ gives a code of size $1$ in exactly $1$ coset. So, 
choosing $\cB^1$ with skeleton code $\{111000,000111\}$ gives $\#\cB^1=q^3+1$ and the other $q^3-1$ cosets for $111000$ give codes with $\#\cB^j=q^3$ for $2\le j\le q^3$. 
Thus, we have
$$
  E_q((6,6),(3,3),6)\ge q^9+2q^3+1
$$ 
and combining Inequality~(\ref{ie_12_6_6_construction}) with Corollary~\ref{cor_construction_1} gives:
\begin{Proposition}
\begin{eqnarray*}
  A_q(12,6;6) &\ge&  q^{24}+q^{15}+q^{14}+2q^{13}+3q^{12}+3q^{11}+3q^{10}+3q^9+q^8\\ &&-q^7-2q^6-3q^5-3q^4-q^3-2q^2-q
\end{eqnarray*}
\end{Proposition} 
Note that $\cB=\cup_{j=1}^{q^3} \cB^j$ has size $q^6+1$, which is not too large compared to the known lower bounds for $A_q(6,4;3)$, see Proposition~\ref{prop_6_4_3}. 

From the general point of view we propose the following challenging research problem. For given parameters $n$, $d$, $d'$, $k$, and $q$ construct a $(n,\star,d,k)$ {\cdc} 
$\cB$ and a partition of $\cB$ into $(n,\star,d',k)$ {\cdc}s $\cB^j$, where $1\le j\le r$ for some integer $r$, such that 
\begin{equation}
  \sum_{j=1}^r \left(\#\cB^j\right)^2\label{eq_target}
\end{equation}
is as large as possible. Provide lower and upper bounds for (\ref{eq_target}).

For $n=6$, $d=4$, $d'=6$, $k=3$, and $q=2$ we have $A_2(6,4;3)=77$ and $A_2(6,6;3)\le 9$ so that the sum in (\ref{eq_target}) is upper bounded by $8\cdot 9^2+5^2\le 673$ while 
our best lower bound is just $1\cdot 9^2+7\cdot 8^2=529$. It is indeed possible to have several subcodes $\cB^j$ of maximum possible cardinality $9$. However, it is unclear if this 
comes at the cost of many subcodes $\cB^j$ with small cardinalities. 

\subsection{Exploiting Lemma~\ref{lemma_special_substructure} for small subspace distances}
\label{subsec_exploit}
While Lemma~\ref{lemma_construction_2} has the advantage that it allows computations in ambient spaces much smaller than the original ambient space, it has the big drawback 
that it is too wasteful if the desired minimum subspace distance is rather small. If we e.g.\ consider lower bounds for $A_q(12,4;6)$ and apply 
Lemma~\ref{lemma_construction_1} and Corollary~\ref{cor_construction_1} and with $\bar{n}=(6,6)$, then suitable choices for $\bar{a}$ in Lemma~\ref{lemma_construction_2} 
are $(2,4)$, $(3,3)$, and $(4,2)$. While we can combine $\bar{a}=(2,4)$ with $\bar{a}=(4,2)$ due to Inequality~(\ref{eq_a_bar_combination}), it turns out, 
see Section~\ref{sec_upper_bounds}, that $E_q((6,6),(2,4),6)$, $E_q((6,6),(4,2),6)$, and $E_q((6,6),(3,3),6)$ all are rather small.    

Given the notation from Lemma~\ref{lemma_special_substructure} the codewords $U$ of the additional subcode $\cD$ only have to satisfy $\dim(U\cap E_i)\ge d/2$ for all $1\le l\le 2$. 
For our chosen parameters it is sufficient if $\dim(U\cap E_1)=\dim(U\cap E_2)=2$, so that $(U\cap E_1)\times (U\cap E_2)$ is only a rather small part of $U$, which allows 
additional freedom. Here we generalize Definition~\ref{def_E_a} to:
\begin{Definition}
  \label{def_E}
  Let $l\ge 2$, $k\ge 1$, $d\ge 2$ with $d\equiv 0\pmod 2$, $\bar{n}=\left(n_1,\dots,n_l\right)\in\N^l$, and $n:=\sum_{i=1}^l n_i$. Set $\sigma_i=\sum_{j=1}^{i} n_j$ for $1 \le i \le l$ 
  and $\sigma_0 = 0$. With this, let $E_i$ denote the $(n-n_i)$-subspace of $\F_q^n$ consisting of all vectors in $\F_q^n$ that have zeroes for the coordinates between $\sigma_{i-1}+1$ 
  and $\sigma_i$ for all $1\le i\le l$. By $E_q(\bar{n},d;k)$ we denote denote the maximum cardinality $M$ of an $(n,M,d,k)_q$ {\cdc} $\cD$ such that every codeword $U\in\cD$ satisfies 
  $\dim(U\cap E_i)\ge d/2$ for $1\le i\le l$.  
\end{Definition}
With this we can state
\begin{eqnarray}
  A_q(12,4;6)&\ge& A_q\!\left(12,4;6;{6\choose 6},{6\choose 0}\right)+A_q\!\left(12,4;6;{6\choose \le 4},{6\choose \ge 2}\right)\notag\\ &&+E_q\big((6,6),4;6\big)\label{ie_12_4_6_construction}. 
\end{eqnarray}
We remark that e.g.\ \cite[Theorem 2.6]{lao2020parameter} gives
$$
  E_2((6,6),4;6)\ge 2154496.
$$ 
Further improvements can e.g.\ be found in \cite{niu2020new}.

\section{Upper bounds}
\label{sec_upper_bounds}

In an $(n,\star,d,k)$ {\cdc} $\cC$ no two codewords can contain the same $(k-d/2+1)$-dimensional subspace $F$, so that 
\begin{equation}
  \label{ie_anticode_bound}
  A_q(n,d;k)\le \frac{\gaussmnum{n}{k-d/2+1}{q}}{\gaussmnum{k}{k-d/2+1}{q}},
\end{equation}
since there are only $\gaussmnum{n}{k-d/2+1}{q}$ such subspaces $F$ and each codeword uses $\gaussmnum{k}{k-d/2+1}{q}$ of them. Inequality~(\ref{ie_anticode_bound}) 
is also known as the \emph{anticode bound}, see e.g.\ \cite{etzion2011error}.

We can refine the argument by counting subspaces per pivot vector. So for $v\in\cG_1(n,k)$ let $\cF$ denote the corresponding Ferrers diagram. 
By $m(q,\cF,1)$ we have denoted the number of $k$-dimensional subspaces $U$ of $\F_q^n$ with pivot vector $v$. Instead of $m(q,\cF,1)$ we also directly 
write $m(q,v,1)$. If $T$ is a $t$-dimensional subspace of $U$, then the pivot vector of $T$ satisfies $v(T)\in\cG_1(n,t)$ and $\supp(v(T))\subseteq \supp(v)$,  
where $\supp(v):=\{1\le i\le n\,:\, v_i\neq 0\}$ denotes the \emph{support} of $v=\left(v_1,\dots,v_n\right)\in\F_2^n$. The $\gaussmnum{k}{t}{q}$ subspaces $T$ 
of $U$ split differently on the different pivot vectors $v'\in\cG_1(n,t)$ with $\supp(v')\subseteq\supp(v)$. Nevertheless the corresponding numbers only 
depend on $v$ and $v'$ so that we denote by $m(q,v,v',1)$ the number of subspaces $T$ of an arbitrary but fixed subspace $U$ of $\F_q^n$ with $p(T)=v'$ and $p(U)=v$. 
If $\supp(v')\not\subseteq \supp(v)$, then $m(q,v,v',1)=0$ by definition. Otherwise we have
\begin{equation}
  m(q,v,v',1)=m(q,\tilde{v},1),  
\end{equation}
where $\tilde{v}$ denotes the restriction of $v'$ to $\supp(v)$. As an example we consider a subspace $U$ with pivot vector $v=(1101100)$. Here we have
\begin{eqnarray*}
  m(q,v,1100000)&=&q^4, \tilde{v}=1100,\\
  m(q,v,1001000)&=&q^3, \tilde{v}=1010, \\
  m(q,v,1000100)&=&q^2, \tilde{v}=1001, \\
  m(q,v,0101000)&=&q^2, \tilde{v}=0110,\\
  m(q,v,0100100)&=&q, \tilde{v}=0101, \text{ and} \\
  m(q,v,0001100)&=&1, \tilde{v}=0011.
\end{eqnarray*}

\begin{Proposition}
  \label{prop_ilp_upper_bound_pivots}
  For $\cV\subseteq \cG_1(n,k)$ we have that $A_q(n,d;k;\cV)$ is upper bounded by the maximum target value of the integer linear program (ILP) 
  maximizing 
  \begin{equation}
    \sum_{v\in\cV} a_v
  \end{equation}
  subject to the constraints
  \begin{equation}
    \sum_{v\in\cV} a_v\cdot m(q,v,v',1) \le m(q,v',1)
  \end{equation}  
  for all $v'\in \cG_1(n,k-d/2+1)$, where $a_v\in\cN$.
\end{Proposition}
\begin{proof}
  Let $\cC$ be a {\cdc} attaining $A_q(n,d;k;\cV)$. By $a_v$ we denote the number of codewords of $\cC$ with pivot vector $v$, so that 
  the target function $\sum_{v\in\cV} a_v$ equals the cardinality $\#\cC$. Since each codeword with pivot vector $v$ contains exactly 
  $m(q,v,v',1)$ $(k-d/2+1)$-dimensional subspaces $T$ with pivot vector $v'$, no two codewords can contain the same such subspace $T$, and there 
  are exactly $m(q,v',1)$ such subspaces in $\F_q^n$, all inequalities for $v'\in \cG_1(n,k-d/2+1)$ are satisfied.
\end{proof}

Of course we can relax the integrality conditions $a_v\in\N$ to $a_v\in\mathbb{R}_{\ge 0}$, in order to obtain a linear program (LP), or add 
additional inequalities $\sum_{v\in\cV'} a_v \le \overline{A}_q(n,d;k;\cV')$ for subsets $\cV'\subseteq\cV$ and known upper bounds $\overline{A}_q(n,d;k;\cV')$ 
for $A_q(n,d;k;\cV')$.

We remark that the special case $\cV=\left({m\choose \le k-d/2},{{n-m}\choose{\ge d/2}}\right)$ of Proposition~\ref{prop_ilp_upper_bound_pivots} was also treated in 
\cite{kurz2020generalized}, where $m\ge k$ is an additional parameter. 

Similar ideas can also be applied to our other descriptions of subcodes. So, let parameters $l$, $\bar{n}$, $\bar{a}$, $d$, and $k=\sum_{i=1}^l a_i$ as in Definition~\ref{def_E_a} be given.
\begin{Proposition}
  Let $\bar{c}=\left(c_1,\dots,c_l\right)\in\N^l$ with $c_i\le a_i$ for $1\le i\le l$ and $\sum_{i=1}^l c_i=k-d/2+1$. Then, we have
  \begin{equation}
    E_q(\bar{n},\bar{a},d)\le \frac{\prod_{i=1}^l \gaussmnum{n_i}{c_i}{q}}{\prod_{i=1}^l \gaussmnum{a_i}{c_i}{q}}.
  \end{equation}
\end{Proposition}
\begin{proof}
  Let the $F_i$, where $1\le i\le l$, as in Definition~\ref{def_E_a} and $F$ be an $(k-d/2+1)$-dimensional subspace of $\F_q^n$ with 
  $\dim(F\cap F_i)=c_i$ for $1\le i\le l$. As observed for the anticode bound, no two codewords can contain the same subspace $F$. 
  Since the total number of such subspaces is given by $\prod_{i=1}^l \gaussmnum{n_i}{c_i}{q}$ and each codeword contains 
  $\prod_{i=1}^l \gaussmnum{a_i}{c_i}{q}$ such subspaces, the upper bound follows. 
\end{proof}
For $l=1$ the statement is equivalent to Inequality~(\ref{ie_anticode_bound}). As an example we consider $\bar{n}=(6,6)$, $\bar{a}=(2,4)$, $d=4$, and $q=2$.
For $\bar{c}=(1,4)$ we obtain
$$
  E_2((6,6),(2,4),4)\le \frac{\gaussmnum{6}{1}{2}\cdot\gaussmnum{6}{4}{2}}{\gaussmnum{2}{1}{2}\cdot \gaussmnum{4}{4}{2}}=13671.
$$
Similarly, we obtain
$$
  E_2((6,6),(4,2),4)\le 13671
$$
and
$$
  E_2((6,6),(3,3),4)\le 129735
$$
for $\bar{c}=(2,3)$. Note that $E_2((6,6),(2,4),4)+E_2((6,6),(4,2),4)+E_2((6,6),(3,3),4)\le 157077$ while $E_2((6,6),4;6)\ge 2154496$.  

For $l=2$ we can also deal with the situation of Definition~\ref{def_E}. Note that the $k$-dimensional codewords $U$ have to intersect the disjoint spaces $E_1$ and $E_2$ in dimensions 
at least $d/2$ each. Thus for each $(k-d/2+1)$-dimensional subspace $T$ of $U$ we have $\dim(T\cap E_1)+\dim(T\cap E_2)\ge d/2+1$, so that:
\begin{Proposition}
  For parameters as in Definition~\ref{def_E} with $l=2$ we have $E_q(\bar{n},d;k) \le$
  $$
    \frac{\#\left\{T\le \F_q^n\,:\, \dim(U)=k-d/2+1, \dim(T\cap E_1)+\dim(T\cap E_2) \ge d/2+1\right\}}{\gaussmnum{k}{k-d/2+1}{q}}.
  $$
\end{Proposition}

We propose it as an open problem to formulate an upper bound for $E_q(\bar{n},d;k)$ similar to the one in Proposition~\ref{prop_ilp_upper_bound_pivots}, i.e., 
to take the different possibilities of the dimensions of the intersections $\dim(U\cap E_i)$ and $\dim(T\cap E_i)$ into account.

As a further line of research we would like to remark that the anticode bound from Inequality~(\ref{ie_anticode_bound}) can be sharpened to the so-called 
\emph{Johnson bound}
\begin{equation}
  \label{ie_johnson}
  A_q(n,d;k)\le \left\lfloor \frac{\left(q^n-1\right)\cdot A_q(n,-1,d;k-1)}{q^k-1}\right\rfloor 
\end{equation}
if $k\ge 2$, see e.g.\ \cite{etzion2011error,xia2009johnson}. If Inequality~(\ref{ie_johnson}) is applied iteratively without rounding down, then we end up with Inequality~(\ref{ie_anticode_bound}), 
see e.g.\ \cite{heinlein2017asymptotic,xia2009johnson}. Using the theory of $q^r$-divisible linear codes over $\F_q$ with respect to the Hamming distance, Inequality~(\ref{ie_johnson}) 
was further tightened in \cite[Theorem 12]{kiermaier2020lengths}. Applied iteratively, it constitutes the tightest known upper bound for 
$A_q(n,d;k)$ when $k<d/2$ and $(q,n,d,k)\neq (2,6,4,3), (2,8,6,4)$. So, the question arises if the underlying ideas of Inequality~(\ref{ie_johnson}) and its tightening in 
\cite[Theorem 12]{kiermaier2020lengths} can also be applied to conclude improved upper bounds for $A_q(n,d;k;\cV)$, $E_q(\bar{n},\bar{a},d)$, and $E_q(\bar{n},d;k)$.   

\section{Open problems for further research}
\label{sec_open_problems}
Here we briefly summarize the open problems mentioned in Section~\ref{sec_improved_constructions} and Section~\ref{sec_upper_bounds}.
\begin{itemize}
  \item Find more examples of generalized skeleton codes as in Proposition~\ref{prop_ex_generalized_skeleton_code} and Proposition~\ref{prop_ex_generalized_skeleton_code2}.
  \item Determine tighter bounds for e.g.\ $E_2\big((5,5),(2,3),4\big)$ and $E_q\big((5,5),(2,3),4\big)$.
  \item Study which {\FDRM} codes allow subcodes with properties as mentioned in Lemma~\ref{lemma_cosets}.
  \item Find improved constructions for {\cdc}s with relatively small subspace distances that can exploit the structural properties mentioned in Lemma~\ref{lemma_special_substructure}. 
  \item Provide more specialized upper bounds for subcodes appearing in constructions for {\cdc}s in the literature, cf.~Section~\ref{sec_upper_bounds}.
\end{itemize}


\end{document}